\documentclass[conference]{IEEEtran}
\IEEEoverridecommandlockouts
\usepackage{cite}
\usepackage{amsmath,amssymb,amsfonts}
\usepackage{amsthm}
\usepackage{mathpartir}
\providecommand{\llbracket}{\mathopen{[\![}}
\providecommand{\rrbracket}{\mathclose{]\!]}}
\usepackage{algorithm}
\usepackage{algpseudocode}
\usepackage{graphicx}
\usepackage{booktabs}
\usepackage{multirow}
\usepackage{array}
\usepackage{xcolor}
\usepackage{listings}
\usepackage{scratch3}
\usepackage{url}
\usepackage[hidelinks]{hyperref}
\usepackage{microtype}
\usepackage{pgfplots}
\pgfplotsset{compat=1.17}
\usetikzlibrary{arrows.meta,positioning,shapes.geometric,shapes.symbols,calc}
\usepackage[most]{tcolorbox}
\usepackage{pifont}

\newcommand{\sys}{StackSwap}

\newcommand{\sched}{\ensuremath{\mathit{Sched}}}
\newcommand{\lensto}{\ensuremath{\sqsubseteq}}

\newcommand{\runin}[1]{\smallskip\noindent\textbf{\textit{#1}}\hspace{0.6em}}
\newcommand{\takeaway}[2]{\smallskip\noindent\textbf{#1}\hspace{0.5em}\textit{#2}\smallskip}
\newtcolorbox{problembox}[1]{colback=gray!6,colframe=black!70,coltitle=white,
  fonttitle=\bfseries\small,title={#1},boxrule=0.5pt,arc=1pt,left=4pt,right=4pt,top=3pt,bottom=3pt}

\newtheorem{theorem}{Theorem}
\newtheorem{definition}{Definition}
\newtheorem{lemma}{Lemma}
\newtheorem{corollary}{Corollary}

\newcommand{\figonescale}{0.62}\newcommand{\figoneleft}{0.36\linewidth}\newcommand{\figoneright}{0.60\linewidth}\newcommand{\figonesep}{\hfill}

\newcommand{\tblfont}{\footnotesize}\newcommand{\tblsep}{3pt}\newcommand{\tblLabelWide}{1.9in}

\newenvironment{figwide}{\begin{figure*}}{\end{figure*}}
\newcommand{\pipelinegraphic}{\includegraphics[width=\textwidth]{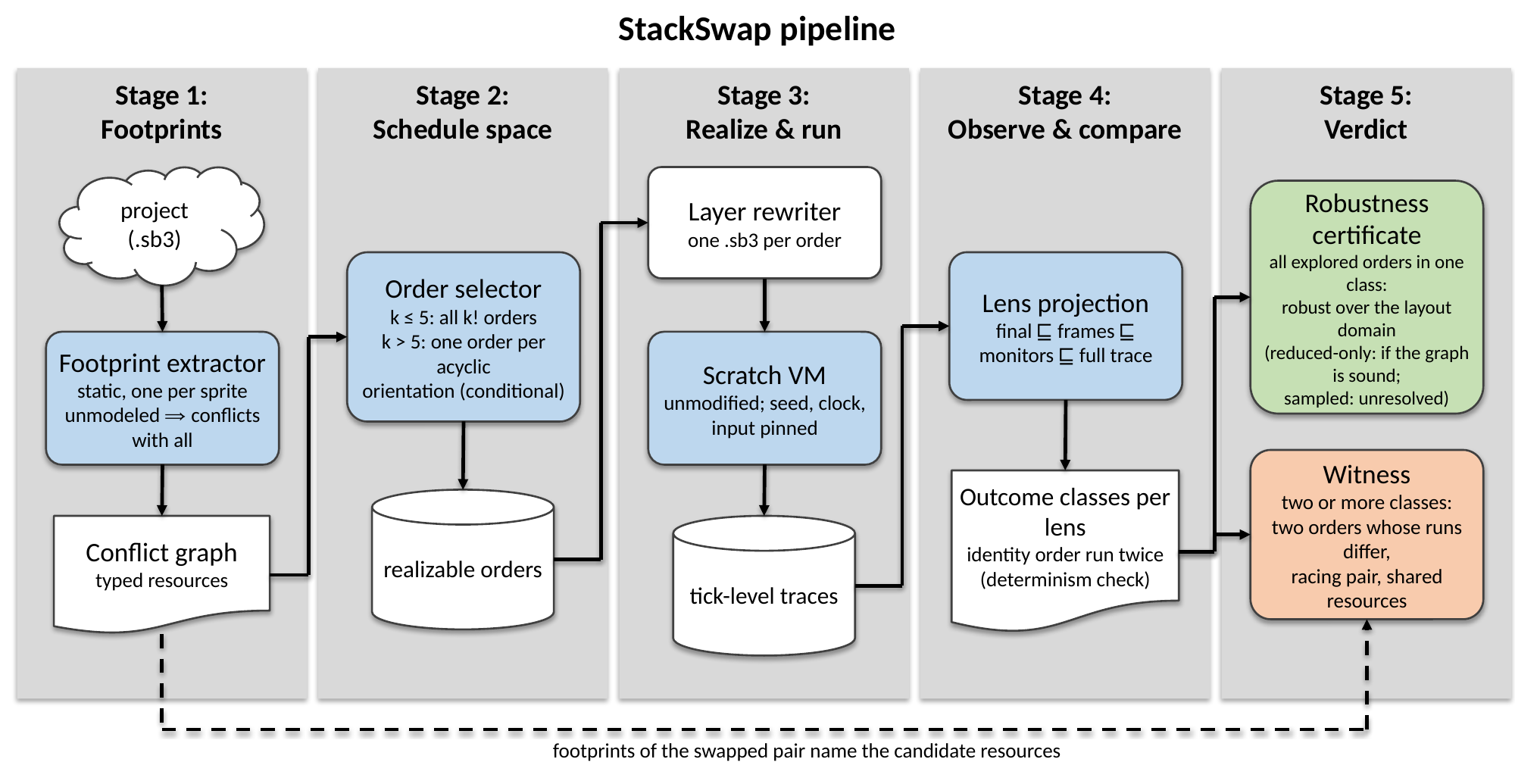}}

\newcommand{\ciCourseSens}{19.4--34.9}
\newcommand{\ciCourseSensNinety}{20.1--35.8}
\newcommand{\ciCourseSensThreeH}{26.8--43.5}
\newcommand{\ciJjcFlip}{0.7--3.9}
\newcommand{\ciPatRandAll}{37.1--62.9}
\newcommand{\ciPublicSens}{12.3--26.0}

\newcommand{\ciSfFlip}{0.0--2.9}
\newcommand{\nBenchPrograms}{46}
\newcommand{\nBinaryCourse}{10}
\newcommand{\nBinaryPublic}{9}
\newcommand{\nCourseAssignments}{13}

\newcommand{\nCourseConc}{121}

\newcommand{\nCourseExact}{109}
\newcommand{\nCourseExactRobust}{87}
\newcommand{\nCourseGraphComplete}{107}
\newcommand{\nCourseGraphMeasured}{121}

\newcommand{\nCourseRedBelow}{13}
\newcommand{\nCourseRedEnum}{101}
\newcommand{\nCourseReducedBelow}{13}

\newcommand{\nCourseReducible}{101}
\newcommand{\nCourseSampled}{12}

\newcommand{\nCourseSampledConcSens}{10}

\newcommand{\nCourseSeedAgree}{117}

\newcommand{\nCourseSens}{32}

\newcommand{\nCourseSensFull}{38}

\newcommand{\nCourseSensNinety}{33}
\newcommand{\nCourseSensSeedThree}{32}
\newcommand{\nCourseSensSeedTwo}{28}
\newcommand{\nCourseSensThreeH}{42}

\newcommand{\nCourseTotal}{145}

\newcommand{\nCourseZeroMiss}{101}

\newcommand{\nDetectCourseN}{22}
\newcommand{\nDetectCourseTwoSprite}{7}
\newcommand{\nDetectPublicN}{13}

\newcommand{\nHorizonCompared}{121}

\newcommand{\nHorizonLost}{0}
\newcommand{\nJjcCovFlip}{7}

\newcommand{\nJjcFlip}{5}

\newcommand{\nJjcKtwo}{299}
\newcommand{\nJjcKtwoUsers}{108}
\newcommand{\nJjcPrematureDone}{6}
\newcommand{\nJjcProblems}{10}
\newcommand{\nJjcProblemsUsed}{7}
\newcommand{\nJjcStudied}{372}
\newcommand{\nJjcStudiedUsers}{125}
\newcommand{\nJjcTotal}{521}
\newcommand{\nJjcUsers}{132}

\newcommand{\nLbNegCourse}{87}

\newcommand{\nLbNegPublic}{83}
\newcommand{\nLbPosCourse}{32}
\newcommand{\nLbPosPublic}{22}

\newcommand{\nLocPrevRows}{28}
\newcommand{\nLocPrevRuns}{172}
\newcommand{\nLocReqAllRows}{5}
\newcommand{\nLocReqDetEdge}{0}
\newcommand{\nLocReqRows}{2}
\newcommand{\nLocReqRuns}{14}
\newcommand{\nLocReqTruthEdge}{5}

\newcommand{\nPatBcastAll}{4}

\newcommand{\nPatCloneAll}{1}

\newcommand{\nPatRandAll}{27}
\newcommand{\nPatRandCourse}{20}
\newcommand{\nPatRandExhAll}{16}
\newcommand{\nPatRandNonexhAll}{11}

\newcommand{\nPatRandPure}{5}
\newcommand{\nPatRandWithCounter}{21}
\newcommand{\nPatRandWithOpaque}{1}
\newcommand{\nPatRandWithTimer}{3}
\newcommand{\nPatSenseAll}{8}

\newcommand{\nPatVarAll}{14}

\newcommand{\nPatVarWeakAll}{2}
\newcommand{\nProgramsStudied}{767}

\newcommand{\nPublicConc}{121}
\newcommand{\nPublicConcNoExtstrip}{109}
\newcommand{\nPublicExact}{97}

\newcommand{\nPublicExtstripConc}{12}
\newcommand{\nPublicExtstripSens}{1}
\newcommand{\nPublicGraphComplete}{43}
\newcommand{\nPublicGraphMeasured}{121}
\newcommand{\nPublicLoaded}{249}

\newcommand{\nPublicReducedBelow}{45}
\newcommand{\nPublicReducedOnlyRobust}{13}
\newcommand{\nPublicReducible}{81}
\newcommand{\nPublicSampled}{24}

\newcommand{\nPublicSampledSens}{8}
\newcommand{\nPublicSens}{22}
\newcommand{\nPublicSensFull}{28}
\newcommand{\nPublicSensNoExtstrip}{21}

\newcommand{\nPublicStubExtstrip}{20}

\newcommand{\nPublicTotal}{250}

\newcommand{\nPublicUntrusted}{1}
\newcommand{\nPublicZeroMiss}{81}

\newcommand{\nRecheckTables}{182}

\newcommand{\nReqCovFlipAll}{17}
\newcommand{\nReqFlipAll}{5}
\newcommand{\nReqFlipLogicAll}{3}
\newcommand{\nReqFlipOtherFail}{3}
\newcommand{\nReqFlipRandomAll}{2}
\newcommand{\nReqIdentFalse}{247}
\newcommand{\nReqIdentTrue}{1504}
\newcommand{\nReqIdentUnknown}{184}
\newcommand{\nReqKtwoAll}{427}
\newcommand{\nReqPredConfirmOnly}{7}
\newcommand{\nReqPredGated}{26}
\newcommand{\nReqPredUnconditional}{50}
\newcommand{\nReqPredicates}{83}
\newcommand{\nReqProblems}{19}

\newcommand{\nRunsTotal}{14,454}
\newcommand{\nSampledAll}{36}

\newcommand{\nSensExhAll}{35}

\newcommand{\nSensNonexhAll}{19}
\newcommand{\nSfCovFlip}{10}

\newcommand{\nSfFlip}{0}

\newcommand{\nSfKtwo}{128}

\newcommand{\nSilentLayerPublic}{7}

\newcommand{\pCourseGraphDensity}{1.00}

\newcommand{\pCourseSens}{26.4}

\newcommand{\pCourseSensThreeH}{34.7}

\newcommand{\pDetectThreeCourse}{97.7}

\newcommand{\pDetectTwoCourse}{91.4}
\newcommand{\pDetectTwoPublic}{79.0}
\newcommand{\pJjcFlip}{1.7}

\newcommand{\pPatLogicAll}{50.0}
\newcommand{\pPatRandAll}{50.0}

\newcommand{\pPublicGraphDensity}{0.70}
\newcommand{\pPublicSens}{18.2}

\newcommand{\pReqFlipAll}{1.2}
\newcommand{\pReqIdentTrue}{85.9}

\newcommand{\pSfFlip}{0.0}
\newcommand{\rJjcFlipGivenFlag}{18}

\newcommand{\rLbPrecAnyCourse}{26}

\newcommand{\rLbPrecAnyExhPublic}{30}
\newcommand{\rLbPrecAnyPublic}{26}
\newcommand{\rLbPrecConcCourse}{26}
\newcommand{\rLbPrecConcPublic}{29}
\newcommand{\rLbRecAnyCourse}{97}
\newcommand{\rLbRecAnyPublic}{95}
\newcommand{\rLbRecConcCourse}{91}
\newcommand{\rLbRecConcPublic}{86}
\newcommand{\rSfFlipGivenFlag}{0}

\newcommand{\xCostCourseMaxS}{81}
\newcommand{\xCostCourseMedianS}{1.2}

\newcommand{\xCostCourseThreeHMedianS}{2.6}

\newcommand{\xCourseRedMedian}{1.0}
\newcommand{\xMedMinClassCourse}{17}

\newcommand{\xMedOutcomesCourse}{5.5}
\newcommand{\xMedOutcomesPublic}{2}
\newcommand{\xMedSavedClassCourse}{17}
\newcommand{\xMedSavedClassPublic}{50}
\newcommand{\xPublicRedMedian}{1.5}

\newcommand{\nLaneCourseNewSens}{32}

\newcommand{\nLanePublicNewSens}{22}

\newcommand{\nLaneCourseNinetyCompared}{145}

\newcommand{\nLaneCourseNinetyChanged}{1}

\newcommand{\nLaneCourseThreeHCompared}{145}

\newcommand{\nLaneCourseThreeHChanged}{1}

\newcommand{\nLaneSfBoth}{128}
\newcommand{\nLaneSfFlipAgree}{128}

\newcommand{\nLaneSfFlipChanged}{0}

\newcommand{\nLaneJjcBoth}{299}
\newcommand{\nLaneJjcFlipAgree}{297}

\newcommand{\nLaneJjcNewFlip}{3}

\newcommand{\nLaneJjcFlipChanged}{2}

\newcommand{\nLaneAllCompared}{393}

\newcommand{\nLaneAllFullCompared}{241}
\newcommand{\nLaneAllFullChanged}{1}

\newcommand{\nLaneReqFlipLost}{2}

\begin{document}

\title{The Invisible Scheduler: Dragging a Sprite\\ Can Change What a Scratch Program Does}

\author{%
\IEEEauthorblockN{Xinyue Feng}
\IEEEauthorblockA{Independent Researcher\\China\\lily.china@outlook.com}
\and
\IEEEauthorblockN{Hanyuan Shi}
\IEEEauthorblockA{Independent Researcher\\China\\shihanyuan1995@gmail.com}
\and
\IEEEauthorblockN{Yuan Si\textsuperscript{*}}
\IEEEauthorblockA{University of Waterloo\\Canada\\y3si@uwaterloo.ca\\\textsuperscript{*}Corresponding author}}

\maketitle

\begin{abstract}
Tens of millions of children program in Scratch, whose programs are concurrent: when the green flag is clicked, every sprite's scripts start together and share the project's state. Which script starts first is decided by the sprites' front-to-back stacking order, which no block reads. Dragging a sprite brings it to the front, and the order is saved with the project. Hence, a program can work on the author's screen and fail on the teacher's, with the same blocks.

Our key observation is that, with the rest of the file fixed, rearranging the sprites that start scripts among their positions produces exactly the permutations of their initial stacking order. Under a fixed input and seed the unmodified virtual machine reproduces every one of them. \sys{} runs a project under these orders, all of them for up to five sprites, one per class of a conflict graph beyond, and a sample where neither is feasible, and compares the runs under four observation lenses. It returns a two-run witness naming an adjacent pair of sprites whose swap changes the outcome, with the resources they share as candidate causes; otherwise an exhaustive or graph-conditional robustness certificate, or an unresolved verdict.

On \nProgramsStudied{} real programs from a course, an online judge, and a random public sample, \pCourseSens\% (course) and \pPublicSens\% (public) of those with a scheduling choice behave differently under some stacking order. Among the sensitive programs whose every order ran, the saved order's outcome recurs under a median \xMedSavedClassCourse\% (course) and \xMedSavedClassPublic\% (public) of the orders. Under scripted play, a predicate written from a stated requirement holds under one order and fails under another for \nReqFlipAll{} of \nReqKtwoAll{} student submissions with a scheduling choice. For half of the sensitive programs the random stream is the racing pair's highest-priority shared resource. We close with recommendations for learners, graders, and the Scratch platform.

\end{abstract}

\section{Introduction}
\label{sec:intro}

Tens of millions of children program in Scratch~\cite{resnick2009scratch,scratchstats}, and for many of them it is the first programming language they meet. The blocks are designed to show everything the computer will do. This paper is about a coordinate the blocks do not show, and about what happens when nobody tells the learner that it exists.

Under the surface a Scratch project is a concurrent program~\cite{maloney2010scratch}. Every sprite carries scripts that start together when the green flag is clicked, and all of them read and write shared state: the project's variables and lists, and the position, costume, and visibility of every sprite. The learner writes interacting threads without ever naming one, and the order in which those threads start is written nowhere in the blocks. It is decided by the front-to-back stacking order of the sprites on the stage (Section~\ref{sec:background}): the virtual machine starts the front-most sprite's scripts first and works its way to the back. The stacking order changes whenever a sprite is dragged, because the editor brings a dragged sprite to the front before it moves it. Whichever sprite the learner last nudged into place runs first, and that order is saved with the project.

Figure~\ref{fig:robot} shows what this looks like in a real program from our study. A student builds a game in which a robot hunts viruses. One sprite's script resets the shared score and spawns the viruses; the robot's script moves under the arrow keys and ends the game once the score reaches four. The file stores the score at four, left there by the last run. In the saved stacking order the virus sprite is in front, so its reset runs first and the game plays normally. Bring the robot forward while tidying the scene, and the robot's script runs first: it reads the stale four before the reset, declares the game won, and stops everything. A teacher who opens a copy saved after such a drag sees a game that ends before it begins; the student sees a game that works. The two saved variants have the same blocks and different stacking orders, and nothing in the language marks the difference.

Three properties make this mechanism more than a curiosity. It is \textbf{invisible}: no block reads the stacking order as a schedule, the editor gives no sign that layers affect execution, and Scratch's documentation describes them as a drawing order. The community wiki records the effect in one sentence~\cite{scratchwikilayer}, and forum threads rediscover it~\cite{scratchforumorder}. It is \textbf{deterministic}: the same file starts its scripts in the same order on every machine and every run, so the start order is stable, reproducible, and fixed at the moment of saving, a coin flipped once and then kept. And it is \textbf{ordinary}: it takes no code to trigger, for example a drag, an added sprite, or a run of a program that itself reorders its sprites. A teacher or an automated grader that meets it may judge a requirement under an order the author never tested.

\begin{figure}[t]
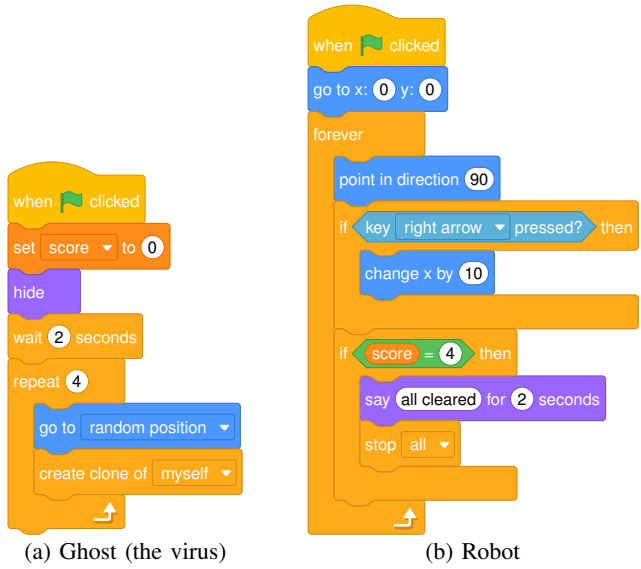

\centering
\begin{minipage}[b]{\figoneleft}
\centering
\begin{scratch}[scale=\figonescale]
\blockinit{when \greenflag clicked}
\blockvariable{set \selectmenu{score} to \ovalnum{0}}
\blocklook{hide}
\blockcontrol{wait \ovalnum{2} seconds}
\blockrepeat{repeat \ovalnum{4}}{
\blockmove{go to \selectmenu{random position}}
\blockcontrol{create clone of \selectmenu{myself}}}
\end{scratch}\par\smallskip
{\small (a) Ghost (the virus)}
\end{minipage}\figonesep%
\begin{minipage}[b]{\figoneright}
\centering
\begin{scratch}[scale=\figonescale]
\blockinit{when \greenflag clicked}
\blockmove{go to x: \ovalnum{0} y: \ovalnum{0}}
\blockinfloop{forever}{
\blockmove{point in direction \ovalnum{90}}
\blockif{if \boolsensing{key \selectmenu{right arrow} pressed?} then}{\blockmove{change x by \ovalnum{10}}}
\blockif{if \booloperator{\ovalvariable{score} = \ovalnum{4}} then}{
\blocklook{say \ovalnum{all cleared} for \ovalnum{2} seconds}
\blockcontrol{stop \selectmenu{all}}}}
\end{scratch}\par\smallskip
{\small (b) Robot}
\end{minipage}
\caption{A real student submission from our study. (a) The virus sprite's script resets the shared \textsf{score} and spawns four clones. (b) The robot's script moves under the arrow keys and stops every script once \textsf{score} reaches four (three key handlers elided). Which of the two runs first is decided by which sprite is in front.}
\label{fig:robot}
\end{figure}

Unfortunately, existing analyses of Scratch programs look elsewhere: test generators such as Whisker~\cite{stahlbauer2019whisker} drive a project with synthesized inputs, static checkers such as LitterBox~\cite{fraser2021litterbox} flag bug patterns in the block graph, and neither varies the schedule. Classical concurrency testing does vary the schedule~\cite{flanagan2005dpor,abdulla2014optimal}. However, it explores the interleavings of a thread set, a model far larger than the one saved coordinate this paper varies.

We take the stacking order itself as the object of study and ask (i) how common it is for a real program's behavior to depend on it, (ii) how severe the dependence is, a rare corner or a coin flip, (iii) whether it changes whether a program does what its assignment asked, and (iv) what, in the code, causes it. Answering these questions requires exploring the schedules a user can realize by rearranging the sprites that start scripts, and our key observation makes that exploration finite: \emph{with the rest of the file fixed, rearranging the sprites that start scripts among their positions produces exactly the permutations of their initial stacking order, and the unmodified VM reproduces every one of them}. \sys{} explores that space: all orders for up to five sprites, one order per class of a dependence-based reduction beyond, and a sample when neither is feasible. The reduction's certificate is conditional on the soundness of its conflict graph. It compares the runs under four observation lenses, from the final state to the full trace. It returns a robustness certificate, a two-run witness that names an adjacent pair of sprites whose swap changes the outcome, with the resources their footprints share as candidate causes, or, when a sample finds no witness, an unresolved verdict. Every witness is a layout a user can have on screen.

We ran \sys{} on \nProgramsStudied{} real programs: a course assignment set, a uniformly random sample of public projects rebuilt with their original assets, and student submissions to an online judge. About one program in four with a scheduling choice in the course set, and one in five in the public sample, is schedule-sensitive at the final-state lens (\pCourseSens\% of \nCourseConc{} and \pPublicSens\% of \nPublicConc{}). Among the sensitive programs whose every order ran, the saved order's behavior recurs under a median \xMedSavedClassCourse\% (course) and \xMedSavedClassPublic\% (public) of the orders, so the dependence is a coin flip or worse rather than a corner case. Encoding each assignment's stated requirements as predicates over the recorded run, we find that for \nReqFlipAll{} of \nReqKtwoAll{} student submissions with a scheduling choice, a predicate holds under one stacking order and fails under another, and for \nReqCovFlipAll{} more the order changes whether a predicate returns a determinate value. A resource-typed attribution labels the certified racing pair of every sensitive program in the prevalence study by the highest-priority resource it shares, and for half of them that resource is the random stream.

This paper makes the following contributions.
\begin{itemize}
\item We identify the mechanism by which sprite stacking order schedules Scratch programs. We prove that the realizable schedule space is the finite set of permutations of the initial stacking order, every element of which the unmodified VM reproduces (Sections~\ref{sec:background} and~\ref{sec:model}).
\item We give a cost theory for checking the property. For a sound conflict graph the observation factors through the graph's acyclic orientations, so one representative run per orientation, $(-1)^{k}\chi_G(-1)$ of them by Stanley's identity, enumerates every outcome. For every graph we construct a project on which no explorer that must return an observed witness for every outcome can use fewer runs (Section~\ref{sec:tool}).
\item We present \sys{}, which explores that space exhaustively, by a graph-conditional reduction, or by sampling, and localizes a sensitivity to an adjacent pair of sprites whose swap changes the outcome (Section~\ref{sec:tool}).
\item We report the first measurement of the phenomenon on real programs: how common it is, how severe, whether it changes the satisfaction of assignment requirements, and what causes it. Random re-execution and a static bug checker serve as points of comparison (Sections~\ref{sec:study} and~\ref{sec:results}).
\item We draw concrete recommendations for learners, for teachers and automated graders, and for the Scratch platform (Section~\ref{sec:implications}).
\end{itemize}

\section{Background and Motivating Example}
\label{sec:background}

\subsection{How Scratch runs a project}

A Scratch project~\cite{maloney2010scratch,resnick2009scratch} is a set of targets: one stage and a number of sprites. Each target owns scripts, and each script begins with a hat block naming the event that starts it: the green flag, the receipt of a broadcast, the creation of a clone, a key press, or a click. A target also owns mutable state that other scripts can observe: local and global variables and lists, and the sprite's position, direction, costume, size, visibility, and graphic effects. Scripts that run at the same time share this state, so a project is a concurrent program over shared memory even though its author only ever manipulates blocks.

The runtime advances in ticks, thirty per second in the editor's default setting. Within a tick the virtual machine (VM) sweeps the active scripts, stepping each until it yields, and repeats the sweep until every script has yielded, a script has requested a redraw, or the frame budget ends. A script yields (i) at a loop boundary, (ii) at a timed wait, (iii) at the join of a broadcast-and-wait, or (iv) at a screen refresh; a loop whose body redraws the stage runs one iteration per tick. Execution is cooperative: a script holds the machine until it chooses to yield. The consequence that matters here is that the first slice of every green-flag script, everything before its first yield, runs in one sweep, one sprite after another, in a fixed order. That sweep decides the introduction's game: the virus script's first slice resets the score before yielding at its wait, the robot script's first slice runs the first iteration of its loop including the score test, and whichever sprite is in front runs first. With the robot in front the test reads the stored four, and the announcement that stops everything begins before the reset.

\subsection{The invisible scheduler: who goes first}

The VM keeps its targets in an ordered list, the stage first and then the sprites from back to front. This list is the sprites' stacking order, the same order that decides which sprite is drawn on top of which. When the green flag fires, the VM walks the list from its end to its beginning and starts each target's green-flag scripts as it passes~\cite{scratchvm}. The front-most sprite therefore runs first, the sprite behind it second, and so on to the back. The same walk starts (i) the receivers of a broadcast, (ii) the scripts of a pressed key, and (iii) the scripts of a click. Hence, the stacking order is a scheduling decision the language leaves to a list the learner never sees as a schedule (the \emph{invisible} property of Section~\ref{sec:intro}).

Three \emph{ordinary} actions rewrite that list, none of which touches a block:

\begin{itemize}
\item \emph{Dragging a sprite on the stage.} The editor brings a dragged sprite to the front before it moves it (``Dragging always brings the target to the front''~\cite{scratchgui}). For example, nudging a character a few pixels to tidy the scene makes it the first to run.
\item \emph{Adding or duplicating a sprite.} A new sprite is placed on top of the existing ones.
\item \emph{Running the project.} \textsf{go to front layer} and \textsf{go forward 1 layers} reorder the list while the project runs, and the reordered list stays in effect afterwards; clicking the green flag does not reset it.
\end{itemize}

When the project is saved, the current list is written into the file as each sprite's layer number. The order a classmate, a teacher, or an automated grader starts with is therefore the order in effect at the moment of saving. Unfortunately, that need not be the order under which the author last tested the program. Run the project, see it work, drag a sprite to tidy the layout, and save: the tested order and the saved order then differ, and only the saved order leaves the author's computer.

\section{Realizable Schedules and Robustness}
\label{sec:model}

In this section we fix what a schedule is, which schedules a user can realize, and what it means for a program to be robust against them. We write $P$ for a project with its stacking order factored out, $\iota$ for a fixed input policy, and $r$ for a fixed random seed. The stacking order re-enters as the order parameter below, so a statement about the same $P$ under two orders speaks about two saved copies of one program.

\subsection{A core calculus}
\label{sec:calculus}

Figure~\ref{fig:semantics} gives the fragment of the machine that decides order. A sprite owns scripts; a script is a hat and a body; a body is a finite sequence of \emph{slices}, the work between two yield points. A slice is an action $a$ with a store transformer $\llbracket a \rrbracket$ and a footprint $\mathrm{fp}(a)$, which records the resources it reads and writes and the events it raises. The graph theorems use only footprints; the commutation theorem of Section~\ref{sec:tool} additionally assumes the adequacy conditions stated there, and the extractor is built to over-approximate accesses.

A configuration carries the layout $L$, a thread queue split into the threads already stepped in this pass ($D$) and those still pending ($T$), the store, and the index $n$ of the current pass. \textsc{Step} runs the first pending thread for one slice. If the slice raises events, they are delivered in the order the slice raises them: a broadcast's receivers are matched against the targets \emph{in layout order} and appended to the tail of the pending list, so a thread started during a pass can still run in that pass, and a clone is placed directly behind its creator in the layout with its start threads, one per clone-start script, appended to the same tail. For example, two broadcasts in one slice queue the receivers of the first before those of the second, as the VM does. This delivery is the one place the layout enters the queue; the observation function may read it as well. A thread that yielded to the frame is carried in $D$ and skipped by \textsc{Pass} until the next \textsc{Tick}; we leave that flag out of the rules, since skipping never moves a slot. When no thread is pending, either the pass repeats (\textsc{Pass}) or the frame closes (\textsc{Tick}): the clock advances, the observation $\mathrm{obs}_\ell(\sigma, L)$ at the lens $\ell$ is emitted as the label of the step, and the pass index returns to one. A frame closes when every thread has yielded to it, when a slice has requested a redraw, or when the pass budget $B$ is spent; under the pinned clock of Section~\ref{sec:tool} that budget is a constant. Like the frame-yield flag, the redraw flag moves no slot and is left out of the configuration. A hat that fires again for a script whose thread is still alive either restarts it in place or leaves it running, according to the hat; neither moves another thread.

\begin{figwide}[t]
\small
\[
\begin{array}{r@{~}c@{~}l@{\qquad}r@{~}c@{~}l}
r &\in& \mathit{Res} \mathrel{::=} \mathsf{var}\,v \mid \mathsf{pose}\,s \mid \mathsf{clones}\,s \mid \mathsf{chan}\,m \mid \mathsf{rand} \mid \mathsf{timer}
&
a &\in& \mathit{Act},\quad \mathrm{fp}(a) = (R, W, Z) \\[1pt]
b &\mathrel{::=}& \epsilon \mid a \cdot b
&
\theta &\mathrel{::=}& \langle s, b \rangle \\[1pt]
h &\mathrel{::=}& \mathsf{flag} \mid \mathsf{key}\,\kappa \mid \mathsf{click} \mid \mathsf{recv}\,m \mid \mathsf{asclone}
&
C &\mathrel{::=}& \langle L,\ D \parallel T,\ \sigma,\ n \rangle
\end{array}
\]
\begin{mathpar}
\inferrule*[left=\textsc{Step}]
  {\mathrm{fp}(a) = (R,W,Z) \\ \sigma' = \llbracket a \rrbracket \sigma \\ (L', T') = \mathrm{deliver}(L,\ T,\ \mathit{ev}(a,\sigma))}
  {\langle L,\ D \parallel \langle s, a\cdot b\rangle \cdot T,\ \sigma,\ n \rangle
   \;\to\;
   \langle L',\ D \cdot \langle s, b\rangle \parallel T',\ \sigma',\ n \rangle}
\and
\inferrule*[left=\textsc{Done}]
  { }
  {\langle L,\ D \parallel \langle s, \epsilon\rangle \cdot T,\ \sigma,\ n \rangle
   \;\to\;
   \langle L,\ D \parallel T,\ \sigma,\ n \rangle}
\and
\inferrule*[left=\textsc{Pass}]
  {n < B \\ \text{no redraw was requested} \\ \text{some thread in } D \text{ yielded to the pass}}
  {\langle L,\ D \parallel \epsilon,\ \sigma,\ n \rangle
   \;\to\;
   \langle L,\ \epsilon \parallel D,\ \sigma,\ n+1 \rangle}
\and
\inferrule*[left=\textsc{Tick}]
  {n = B \;\lor\; \text{a redraw was requested} \;\lor\; \text{every thread in } D \text{ yielded to the frame} \\ o = \mathrm{obs}_\ell(\sigma, L)}
  {\langle L,\ D \parallel \epsilon,\ \sigma,\ n \rangle
   \;\overset{o}{\Longrightarrow}\;
   \langle L,\ \epsilon \parallel D,\ \sigma,\ 1 \rangle}
\end{mathpar}
\[
\begin{array}{@{}l@{}}
\mathrm{deliver}(L, T, \epsilon) = (L, T) \qquad
\mathrm{deliver}(L, T, \mathsf{send}\,m \cdot E) = \mathrm{deliver}\bigl(L,\ T \cdot \mathit{recv}_L(m),\ E\bigr) \\[1pt]
\mathrm{deliver}(L, T, \mathsf{clone}\,s\,\mathsf{as}\,c \cdot E) = \mathrm{deliver}\bigl(L[\,c \text{ directly behind } s\,],\ T \cdot \mathit{start}(c),\ E\bigr)
\end{array}
\]
\caption{The core calculus. A slice is an action whose footprint records the resources it reads and writes and the events it may raise; it ends at a yield, to the pass or to the frame, and \textsc{Tick} emits the observation as its label. The layout enters the queue in one place, the delivery inside \textsc{Step}, which delivers the events $\mathit{ev}(a,\sigma)$ the slice actually raises, each in $Z$, in the order it raises them: a broadcast appends its receivers $\mathit{recv}_L(m)$, one slot per matching script of each target in layout order, and a clone joins the layout directly behind its creator and its start threads $\mathit{start}(c)$, one per clone-start script in saved order, join the tail. New threads join the tail of the pending list, where they may run in the pass that started them; the pass index $n$ runs from $1$ to the budget $B$. The fragment fixes the queue spine of the VM and leaves out restarts in place, waits, and cancellation, which move no thread.}
\label{fig:semantics}
\end{figwide}

\subsection{The realizable schedule space}

With the seed and the clock pinned, $\llbracket a \rrbracket$ is a function, and the rules are deterministic once the layout is fixed (the \emph{deterministic} property). A run over a horizon of $H$ ticks is therefore a function of the program, the input, the seed, and the initial layout.

\runin{The fixed experiment} Fix the harness environment (the pinned clock, seed, and input, the sensing configuration, and the observation function), a horizon $H$, and a lens $\ell$. Let $\pi$ range over the permutations of the $k$ \emph{contributing} sprites, the sprites whose hats fire under the fixed input or that are reached from one through a broadcast or clone edge. Section~\ref{sec:tool} computes a static superset of this set that does not depend on the layout. The stage and every other sprite are held in their saved places. Write $\mathrm{exec}_H(P,\iota,r,\pi)$ for the run under initial layout $\pi$ and $F(\pi)$ for its observation at $\ell$; the layouts $S_k$ are the parameter domain, the runs are its image, and the outcome partition $\ker F$ is a further quotient. Everything the paper measures is a statement about $F$.

\begin{definition}[Admissible layouts]
Fix the $k$ contributing sprites and the $k$ layer positions they occupy in the saved file. An admissible layout permutes these sprites among those positions; every other target, the script order inside each sprite, the assets, and the saved state stay fixed.
\end{definition}

In the layer coordinate, then, the one free choice is the initial layout. One might ask why a sprite that never starts a script (a \emph{silent} sprite) is held in place rather than permuted. It changes no delivery. However, it can still matter: it can occlude a colour test in the pixel configuration, and a contributing sprite that moves by a relative number of layers counts it. We hold such sprites in their saved place, so the domain is a restriction of what a user can arrange. \nSilentLayerPublic{} public projects and no course program have a contributing sprite with a relative layer block beside a silent sprite, all of them in the sampled stratum (Section~\ref{sec:results}). Adjacency below refers to the order of the contributing sprites; a silent sprite stacked between two of them does not separate them. The realizable schedule space is
\[
  \sched(P,\iota,r,H) \;=\; \{\, \mathrm{exec}_H(P,\iota,r,\pi) \;:\; \pi \in S_k \,\}.
\]
The space is finite. Every element is a run of the unmodified VM package in the fixed environment. Lemma~\ref{lem:layer} establishes the queue invariant everything below rests on: surviving slots keep their insertion order, and each delivery appends its slots in layout order.

\begin{lemma}[Queue-slot discipline]
\label{lem:layer}
Give every thread the queue slot it was appended into; a hat that restarts a live script replaces the slot's payload and keeps the slot, and a hat that leaves a live script running changes nothing. Surviving slots stay in insertion order. If one delivery appends slots for distinct targets $t \neq t'$, then while both survive the slot of $t$ precedes the slot of $t'$ exactly when $t$ preceded $t'$ in the layout at that delivery. Slots from different deliveries are ordered by delivery, which the run itself fixes.
\end{lemma}

\begin{proof}
A delivery appends one slot per matching hat, walking the targets in layout order and the scripts of each target in saved order, so the claim holds when the slots are appended. Every rule preserves the order of surviving slots in $D \cdot T$. \textsc{Step} moves the head of $T$ to the end of $D$ and appends fresh slots after everything present, one delivery after another in the order the slice raised them; \textsc{Done} deletes a slot; \textsc{Pass} and \textsc{Tick} rotate $D$ back in its existing order. Replacement changes a payload, not a position, and a clone takes part in later deliveries from its place in $L$, directly behind its creator.
\end{proof}

\begin{theorem}[Realizability]
\label{thm:realizable}
Every $\pi \in S_k$ has a saved encoding among the admissible layouts, and every admissible layout reads back as one $\pi$. If every admissible layout has a deterministic run that reaches $H$, then $\sched(P,\iota,r,H)$ is exactly the set of runs of these encodings. Any difference between two explored schedules is then exhibited by two runs, in the fixed environment, of saved variants that differ only in the layer coordinate of the contributing sprites.
\end{theorem}

\begin{proof}
A saved project stores a layer number per sprite and the loader sorts the targets on it. Assigning the contributing sprites the layer numbers of their designated positions in the order $\pi$ realizes $\pi$ and moves no other target. Reading the contributing sprites off any admissible layout gives one permutation, so the two directions are inverse. Under the two premises, running an encoding is a function of $\pi$, which defines $\mathrm{exec}$ and makes the two descriptions of the space coincide. The harness fixes the environment and checks completion and repeated-identity agreement (Section~\ref{sec:tool}); determinism itself remains a premise. The fixed seed makes $\mathsf{rand}$ a sequence consumed in slice order, the virtual clock makes $\mathsf{timer}$ a function of the tick index and the last reset, and the input is a fixed deterministic policy: scheduled events and one answer per request. A run that does not reach $H$ is reported rather than compared.
\end{proof}

The theorem is an encoding theorem about the initial layer coordinate; a layer block executed during a run is a fixed action of the program. The calculus of Figure~\ref{fig:semantics} is the explanatory fragment behind the proofs, not a verified model of the whole VM.

\subsection{Observation lenses}

What counts as a difference depends on what an observer watches. We use four observation lenses that form a chain from coarse to fine~\cite{scratchlens}, each a function of the recorded trace, stated here as implemented. The \emph{final-state} lens reports, at the horizon, every variable and the full contents of every list of every target (stage, sprites, and clones), the pose of every sprite and clone (position, direction, costume, size, visibility, every graphic effect, and the speech bubble), and the stage's backdrop and effects. The \emph{frame} lens adds the per-tick poses of every sprite and clone and the stage's backdrop and effects, the \emph{monitor} lens the per-tick variable and list values, and the \emph{full-trace} lens the per-tick population of clone identities. Entities are keyed by identity, a sprite by its name and a clone by its creator and creation ordinal. An observation is therefore a set of per-entity histories rather than a global interleaving log, and reordering work that touches disjoint entities leaves it unchanged. Numbers are quantized to six significant digits before comparison. The quantization absorbs the low-order noise of floating-point sums taken in different orders, at the price of merging differences below that resolution. Each coarser observation is a projection of the next finer one,
\[
  \text{final} \;\lensto\; \text{frame} \;\lensto\; \text{monitor} \;\lensto\; \text{full},
\]
so agreement at a finer lens forces agreement at every coarser one.

\begin{lemma}[Monotonicity]
\label{lem:mono}
If two runs agree under a finer lens, they agree under every coarser lens; a program robust at a finer lens is robust at every coarser one.
\end{lemma}

\begin{proof}
Each lens is a projection of the next finer one, so applying the projection to two equal observations gives equal observations; robustness is agreement over every pair in the space and transfers the same way.
\end{proof}

The same argument says when a verdict transfers across horizons. The monitor and full-trace observations are per-tick histories, so when the input and the execution prefix do not depend on the horizon they are projections of longer ones by truncation, and a witness at a shorter horizon remains one at a longer horizon. Note that the final-state and frame observations include endpoint values, so they are not: two runs can diverge and coincide again. The lens is part of a verdict because it states the strength of the claim. A difference at the full-trace lens may be invisible to a player; a difference at the final-state lens changes a recorded endpoint value, the strongest sensitivity claim, and still not by itself a violated requirement.

\subsection{Schedule-robustness}

\begin{definition}[Schedule-robustness]
$P$ is schedule-robust at lens $L$ for $(\iota,r)$ over horizon $H$ when all schedules in $\sched(P,\iota,r,H)$ agree under $L$. $P$ is schedule-sensitive at $L$ when two of them disagree; the two orders and their observations form a witness.
\end{definition}

\runin{The checking task} Given a project $P$, an input policy $\iota$, a seed $r$, a horizon $H$, and a lens $L$, decide whether all schedules in $\sched(P,\iota,r,H)$ agree under $L$. If they do, return a certificate; if not, return two stacking orders whose runs disagree, together with an adjacent pair of sprites whose swap changes the outcome and the resources their footprints share. Robustness is monotone along the chain by Lemma~\ref{lem:mono}, so a verdict is stated at a grading lens and implies the coarser ones. In every study in this paper the headline lens is the final-state lens, the most conservative choice: a program counted as sensitive there ends in a different state, not only in a different order of intermediate frames. For example, the robot game of Figure~\ref{fig:robot} is sensitive at this lens (at the horizon one order shows the robot's announcement and the other does not), and hence at every finer one.

Schedule-robustness is a property of a set of runs, not of any one run. For a fixed horizon, let $\mathcal{R}_{\ell,H}$ be the set of trace sets whose members all have the same observation at $\ell$ through tick $H$. The robustness question is whether $\sched(P,\iota,r,H) \in \mathcal{R}_{\ell,H}$, a hyperproperty in the sense of Clarkson and Schneider~\cite{clarkson2008hyper}. It is a $2$-safety hyperproperty: a violation is two traces with unequal observations, their prefixes through $H$ fix those observations, and every trace set extending both prefixes violates it, whereas no single prefix does. Two is the witness arity whatever the number of sprites.

\section{\sys{}: Exploring the Realizable Schedules}
\label{sec:tool}

\sys{} answers one question about a project: does the order of its contributing sprites change what the project does, and if so, which adjacent pair of sprites decides it and what do they contend for? It does so by running the project under realizable orders and comparing what it observes (Figure~\ref{fig:workflow}). This section describes (i) the run harness, (ii) the exhaustive exploration and its certificate, (iii) the reduction used beyond five sprites, and (iv) the attribution.

\begin{figwide}[t]
\centering
\pipelinegraphic
\caption{The \sys{} pipeline. Stage~1 extracts a typed footprint per sprite and builds the conflict graph. Stage~2 chooses the orders to run: all $k!$ for up to five sprites, otherwise one per acyclic orientation of the graph, or a sample of five when the graph admits too many. Stage~3 realizes each order by rewriting the layer numbers of the same file and runs it on the unmodified Scratch VM with the seed, clock, and input pinned. Stage~4 projects every trace to the four lenses and groups the orders into outcome classes. Stage~5 returns a certificate, a two-run witness with the pair of sprites whose swap changes the outcome, or an unresolved verdict for a sampled project without a witness.}
\label{fig:workflow}
\end{figwide}

\subsection{Realize, run, observe}

\sys{} runs the production Scratch VM (scratch-vm~5.0.300~\cite{scratchvm}) unmodified, inside a headless harness. To realize an order $\pi$ it rewrites the project's layer numbers so that the sprites stack in that order, exactly as a sequence of drags would, and runs the rewritten project. Nothing else in the file changes, so a witness is two files a user could have saved. A file whose saved layer numbers are invalid is first renumbered in its saved order, which changes no position; one file in the study needed this. The harness has two configurations. The \emph{bounding-box} configuration attaches a renderer stand-in that implements the VM's geometry calls (drawable order, bounds, and fencing) from the costume assets and leaves the VM's own redraw and visibility bookkeeping in place, so the sequencer runs the passes a rendered VM would. It answers touch queries from the costumes' bounding boxes and colour queries with \emph{no}, the approximation a headless analysis makes without a renderer. The \emph{pixel} configuration attaches a headless 2D renderer that answers the same queries from the costumes' pixel silhouettes with Scratch's own colour tolerance. It refuses a project whose costumes it cannot rasterize, and it stops rather than guess when a colour-touch test is polled under a graphic effect it cannot decide. The harness pins every other source of variation. (i) It seeds the random generator, so a draw is a function of the number of prior draws. (ii) It advances the clock one tick of simulated time per step, so a timer read depends on the tick and the last reset. (iii) It delivers the input under a fixed policy: green-flag, key, and click events at fixed ticks, and a text answer whenever the program asks. It also re-implements the two speech-bubble timers, which the VM runs on a real timer, on the same virtual clock, and drains the completions of promise-returning blocks at every tick boundary. Each order runs in a fresh process.

\runin{Contributing sprites} The contributing set is a closure computed from the block graph, independent of the layout. Its roots are the green-flag hats, the hats of the keys and clicks the input policy delivers, and the timer-threshold hats. Its edges follow broadcasts to their receivers, backdrop switches to the scripts they start, and clone creation to the clone-start scripts of the created sprite. A computed message name or backdrop selector reaches every receiver. A sprite reached by no root or edge is silent and stays in its saved place.

One pin removes a machine dependence rather than a source of noise. The VM's sequencer repeats its pass over the threads until they all yield, a redraw is requested, or a work budget of $0.75$ of the step time is spent. It measures that budget on the wall clock~\cite{scratchvm}. On a slower host fewer passes fit in a frame, so host speed changes what a program does. Our harness redirects that timer to the virtual clock, which advances a fixed $0.5$~ms per call. The budget becomes a deterministic cap of fifty passes per frame, the constant $B$ of Figure~\ref{fig:semantics}, and our numbers say nothing about host speed.

We record the observation at every lens of Section~\ref{sec:model} from one trace. In every experiment reported here we run the identity order twice before any comparison and require the two runs to agree at the full-trace lens. A project that fails this diagnostic, or a run that stops before the horizon because a block threw or a bound fired, is reported as untrusted rather than compared. Two exceptions arose in the whole study: one public project (\nPublicUntrusted{} in Table~\ref{tab:prevalence}) fails inside the VM at its first tick under every order, and one more stops the pixel configuration at a colour-touch test it will not decide. With the pins in place, and determinism as the premise this diagnostic checks but does not prove, a difference between two runs traces to the order.

\subsection{Exhaustive exploration and the certificate}

For a project whose contributing sprites number $k \le 5$, we run all $k!$ orders. The verdict is then exact. The project is robust at a lens if and only if all runs agree at that lens. A sensitive verdict comes with the set of distinct outcomes, each with the orders that produce it. A robust verdict is a certificate rather than an absence of observed failures, because Theorem~\ref{thm:realizable} covers every layout in the admissible domain: no permutation of the contributing sprites among their saved positions changes the observation under this fixed experiment. The cost is at most $5! = 120$ layouts, plus the two identity runs, of a program that runs for a few seconds.

\subsection{Beyond five sprites: a conditional reduction}

Beyond five sprites the enumeration grows quickly (720 runs at six, 5{,}040 at seven). However, most of those runs agree: two sprites whose scripts touch disjoint state, in the sense made precise below, cannot influence each other, so swapping them cannot change the observation. We make this precise with a \emph{footprint} for each sprite, a typed record of the resources its scripts read, write, create, delete, and consume: (i) variables and lists, (ii) the pose properties of each sprite, (iii) clone families, (iv) broadcast channels, and (v) the ordered tokens of the random stream and the timer. The footprint includes the sprite's clone-start scripts, is computed statically, and over-approximates by design. An unmodeled opcode is recorded as touching an external resource that everything conflicts with, so it forces dependence and never a false independence. Two sprites conflict when one writes a resource the other reads or writes. Two resources are not variables at all. One is the tail of the thread queue. A sprite that broadcasts or creates a clone appends threads there, and the order in which two sprites do so fixes the relative order of those threads for as long as both survive (Lemma~\ref{lem:layer}). Two starters therefore conflict unless the targets their threads and descendants run in neither overlap nor conflict, and a sender whose message name is computed at run time conflicts with every other starter. The other is the VM's cap of 300 live clones, which every clone creator draws on.

Write $G_P$ for the \emph{conflict graph} (Stage~1 of Figure~\ref{fig:workflow}): one vertex per contributing sprite, one edge per conflicting pair. Three kinds of statement follow: (i) combinatorics about graphs and layouts, which assumes nothing about the VM; (ii) a commutation theorem for an ideal fragment, proved outright; and (iii) the general commutation claim, stated under a named condition on the footprint abstraction, whose consequence, homogeneity of the outcome classes, we check on every exhaustive table.

\runin{Combinatorics} For a layout $\pi$, let $\mathrm{orient}(\pi)$ be the acyclic orientation of $G$ that points every edge from the sprite $\pi$ puts in front to the one it puts behind. Write $\Omega(G)$ for the acyclic orientations of $G$ and $a(G) = |\Omega(G)|$ (at most $k!$, and much less for a sparse graph).

\begin{lemma}[Orientation-class connectivity]
\label{lem:connect}
Two layouts have the same $G$-orientation if and only if a sequence of swaps of adjacent sprites that are not an edge of $G$ joins them.
\end{lemma}

\begin{proof}
An adjacent swap changes the relative order of its two sprites only, so a swap outside $G$ preserves every edge direction. Conversely, the layouts with orientation $\omega$ are the linear extensions of its transitive closure $\leq_\omega$. Let $x$ be first in the target extension; every $y$ before $x$ in the current one is incomparable with $x$, since either comparison would contradict one of the two extensions. Bubble $x$ to the front, delete it, and repeat; incomparable sprites share no edge, so every swap is outside $G$.
\end{proof}

\begin{theorem}[Factorization]
\label{thm:factor}
Let $F : S_k \to \mathcal{O}$ be unchanged by every swap of adjacent sprites outside $G$. Then there is a unique $\overline{F} : \Omega(G) \to \mathcal{O}$ with $F = \overline{F} \circ \mathrm{orient}$, and one layout per orientation observes every outcome $F$ takes.
\end{theorem}

\begin{proof}
Lemma~\ref{lem:connect} makes $F$ constant on each orientation class, so $\overline{F}$ is well defined on the classes that occur. Every acyclic orientation occurs: remove a source repeatedly (a finite acyclic orientation has one, else following incoming edges would close a cycle), and the removal order is a layout with that orientation. Hence every class is nonempty, $\overline{F}$ is unique, and any one member of each class observes its outcome.
\end{proof}

\begin{corollary}[Cost]
\label{cor:cost}
Under Factorization, at most $a(G)$ representative runs enumerate every outcome, and $a(G) = (-1)^{k}\chi_G(-1)$ for the chromatic polynomial $\chi_G$~\cite{stanley1973acyclic}. The count is $1$ for an edgeless graph, $2^{m}$ for a forest with $m$ edges, and $k!$ for the complete graph.
\end{corollary}

\begin{proof}
The bound is Theorem~\ref{thm:factor}; the identity is Stanley's, by the deletion--contraction recurrence $a(G) = a(G-e) + a(G/e)$ that $(-1)^{k}\chi_G(-1)$ also obeys. Every orientation of a forest is acyclic; an acyclic orientation of $K_k$ is a linear order.
\end{proof}

For example, the robot and virus of Figure~\ref{fig:robot} share \textsf{score}, so their graph is one edge and $a(G) = 2$: two runs decide the program.

\runin{The semantic graph} The exhaustive tables (\nRecheckTables{} of them, Section~\ref{sec:results}) let us ask how far the static graph is from the best graph any analysis could produce. Order partitions of $S_k$ by refinement ($Q \preceq Q'$ when every block of $Q$ lies in a block of $Q'$) and graphs on the sprites by edge inclusion. Let $\gamma(G)$ be the partition of $S_k$ by $G$-orientation, and for a partition $Q$ let $\alpha(Q)$ contain a pair of sprites exactly when some layout that places them adjacently changes its $Q$-block when they are swapped. For an experiment, $Q$ is the outcome partition $\ker F$.

\begin{theorem}[The least sound graph]
\label{thm:alpha}
For every partition $Q$ and graph $G$, $\alpha(Q) \subseteq G$ if and only if $\gamma(G) \preceq Q$; moreover $\alpha(\gamma(G)) = G$. Hence $\alpha(Q)$ is the least graph whose orientation classes are homogeneous for $Q$, and $\alpha$ is left adjoint to $\gamma$ under reverse refinement.
\end{theorem}

\begin{proof}
If $\alpha(Q) \subseteq G$, Lemma~\ref{lem:connect} joins any two layouts of a $\gamma(G)$-block by swaps outside $G$, none in $\alpha(Q)$, so all stay in one $Q$-block. If $\gamma(G) \preceq Q$, a swap outside $G$ keeps its orientation block and so its $Q$-block, hence every edge of $\alpha(Q)$ lies in $G$. A swap outside $G$ keeps the orientation class, while swapping the adjacent endpoints of an edge reverses it, which gives $\alpha(\gamma(G)) = G$.
\end{proof}

Hence, if the classes of $G$ are homogeneous for the outcome partition $Q$, then $|Q| \leq a(\alpha(Q)) \leq a(G) \leq k!$. The first gap, $a(\alpha(Q))/|Q|$, is the price of describing outcomes by pairwise orders at all (for example, partition layouts by parity, and $\alpha(Q)$ is complete); the second, $a(G)/a(\alpha(Q))$, is the price of the particular analysis. Homogeneity is decidable from a complete table, which is how Section~\ref{sec:results} checks every enumerable program.

\runin{An ideal fragment} Consider the fragment without deliveries, clones, restarts, and cancellation, in which a fixed cohort of threads runs slices with \emph{local} footprints. A slice's enabledness and effect depend only on the store restricted to $R$, it writes only $W$, and beyond that store effect it alters only its own control and yield state. Let the scheduler run the eligible threads in complete rounds under a fixed round budget, skipping frame-yielded threads until the next tick, and let slices neither read the layout nor change that policy. There the hypothesis of Theorem~\ref{thm:factor} holds outright. Two slices with $W_a \cap (R_b \cup W_b) = \emptyset$ and $W_b \cap R_a = \emptyset$ reach the same store and residual controls in either order; an induction over rounds exchanges the two adjacent blocks of threads one pair of slices at a time while preserving every round boundary, so ticks observe equal states.

\runin{The general case} In contrast, real programs deliver broadcasts, create and delete clones, restart scripts, and wait. What the footprint abstraction must supply for commutation to hold in general has three parts. \emph{Access coverage}: every store effect, enabledness test, delivered event, control decision, and yield of a slice depends only on covered reads and lands only on covered writes; ordered consumption counts as a write, and the scheduling resources (activation and wait state, restart, clone allocation and the clone cap, timer accounting) are resources like any other. \emph{Persistent enabledness}: a pair certified independent stays independent after either side runs. Call two states \emph{related} when they agree on the store, the controls, the input position, and the clock, up to a renaming of fresh identities, and their queues differ only by the permitted swaps. \emph{Continuation compatibility}: from related states the two executions can be matched segment by segment up to the permitted swaps and a renaming of fresh identities, no swap moves work across a tick boundary, and the lens is invariant under the matching. Call an abstraction with these three properties \emph{adequate} for $G$ and $\ell$.

\begin{theorem}[Commutation]
\label{thm:commute}
If every admissible layout runs deterministically to $H$ and the footprint abstraction is adequate for $G_P$ and $\ell$, then swapping two adjacent sprites that are not an edge of $G_P$ leaves $F$ unchanged.
\end{theorem}

\begin{proof}
The two runs start in related states, since the two layouts differ only in the order of the swapped pair while the store, the controls, and the input are the same. Continuation compatibility gives matching finite segments that preserve the relation and emit equal tick labels. Access coverage and persistent enabledness supply the commuting diamonds those segments are built from, one per pair of exchanged slices, with the same store after both sides of a diamond (not after each single slice). Determinism and completion rule out an unbounded delay before a required tick. Induction over matched ticks gives equal records through $H$, and the lens is a function of the record.
\end{proof}

Adequacy relates the abstraction, the program, and the lens. The three rules of the conflict definition (the queue tail, the clone cap, and the computed message name) discharge particular obligations, for instance that two starters with disjoint, non-conflicting targets append threads that commute. They do not prove the extractor adequate for every program, and we do not claim it. What we can check is each instance: by Theorem~\ref{thm:alpha}, homogeneity of a complete table is $\alpha(\ker F) \subseteq G_P$.

\runin{A lower bound} One might wonder whether the bound of Corollary~\ref{cor:cost} is an artifact of a coarse independence test. However, for explorers that must return an observed witness for every outcome, the interface \sys{} implements, it is tight. An \emph{explorer} receives an opaque program, the layout domain, the fixed experiment, and a sound graph $G$. It may request a completed run at any layout and observe $F$ there, and nothing else. It must return a map from every outcome in $F(S_k)$ to a layout at which it observed that outcome.

\begin{theorem}[Optimal witness-producing exploration]
\label{thm:tight}
If $F$ is constant on the classes of $G$, one query per acyclic orientation suffices. For every graph $G$ on $k$ vertices there is a program with conflict graph $G$ and $a(G)$ distinct outcomes at the final-state lens, on which every correct explorer makes at least $a(G)$ queries.
\end{theorem}

\begin{proof}
Sufficiency is Theorem~\ref{thm:factor}. For the bound, give every vertex a sprite and every edge $e = \{u,v\}$ ($u < v$) a variable $v_e$; sprite $u$ assigns $v_e \gets 0$ and sprite $v$ assigns $v_e \gets 1$, each in one initial green-flag slice and with no other blocks. Exact footprints share a variable exactly on the edges, so the conflict graph is $G$. The initial delivery runs the slices in layout order (Lemma~\ref{lem:layer}), so the final value of $v_e$ names the endpoint that ran later and the vector of edge variables is the orientation. The $a(G)$ orientations give $a(G)$ outcomes, each of which a correct explorer must have observed.
\end{proof}

The bound concerns enumerating outcomes with observed witnesses on a worst-case member of the graph class; it says nothing about deciding robustness, which two well-chosen runs settle on that program.

Corollary~\ref{cor:cost} also predicts where a reduction can and cannot help. In the course set the conflict graph is complete for \nCourseGraphComplete{} of the \nCourseGraphMeasured{} concurrent programs (median edge density \pCourseGraphDensity{}), so $a(G_P) = k!$ and no reduction is possible; we measure a saving on \nCourseRedBelow{} of \nCourseRedEnum{} enumerable programs, and the median factor over all of them is \xCourseRedMedian{}. In contrast, the public sample is sparser (median density \pPublicGraphDensity{}, complete for \nPublicGraphComplete{} of \nPublicGraphMeasured{}), and there the same reduction saves a median factor of \xPublicRedMedian{} over all enumerable projects.

Theorem~\ref{thm:commute} is conditional, and a footprint that omits an access would turn a real dependence into a missed outcome. Hence, \sys{} checks each instance it can. For every project with $k \le 5$ the pipeline runs all $k!$ layouts and checks that the representatives cover the outcome set at all four lenses. A separate audit groups the layouts by their orientation of $G_P$ and requires every group to carry one outcome at the final-state lens, on all \nRecheckTables{} enumerated 30-tick tables (Section~\ref{sec:results}); the complete-graph tables of the reduced branch are homogeneous trivially, one layout per class. This certifies homogeneity for those experiments, not the adequacy of the extractor elsewhere. Where the reduction is the only exploration ($k > 5$), the verdict is conditional on the graph being sound for that experiment. A project whose graph admits more than 1{,}024 orientations is run under a sample of five orders, which can exhibit sensitivity but cannot certify robustness, and is reported as a lower bound.

\subsection{Witnesses and attribution}

We attach to every sensitive verdict a witness, two orders and their observations, and an attribution. In the exhaustive table we read the attribution off directly. Two orders that differ only by swapping two adjacent sprites and produce different outcomes make those two sprites a \emph{racing pair}, and we intersect their footprints to name the resources they contend for. Beyond the exhaustive range a witness is two layouts, possibly far apart, and a pair they invert is only a candidate. The following procedure certifies one; it re-runs both witness layouts first and stops if they agree, so every certified pair rests on a reproduced witness.

\begin{theorem}[Adjacent witness localization]
\label{thm:localize}
Let $d(\pi,\pi')$ count the pairs of sprites the two layouts order differently. If $F(\pi) \neq F(\pi')$, there is a path $\pi = \pi_0, \dots, \pi_D = \pi'$ of adjacent swaps with $D = d(\pi,\pi') \leq k(k-1)/2$, and at most $\lceil \log_2 D \rceil$ further runs find a $j$ with $F(\pi_j) \neq F(\pi_{j+1})$. If that swap exchanges $s$ and $t$, then $\{s,t\} \in \alpha(\ker F)$, so every sound graph for the experiment contains the edge.
\end{theorem}

\begin{proof}
Read $\pi'$ left to right and move each sprite to its place by adjacent swaps; each swap corrects exactly one inverted pair, giving $D$ swaps, and no shorter path exists since a swap changes one pair. Keep $u < v$ with $F(\pi_u) \neq F(\pi_v)$; query the midpoint $m$ and keep the half whose endpoints still differ, which exists even if the midpoint shows a third outcome. The interval halves each time. The last claim is the definition of $\alpha$ and Theorem~\ref{thm:alpha}.
\end{proof}

The localized pair certifies an intervention in context, one of possibly several racing pairs: swapping these two sprites in this layout changes the observation. Note that it does not say the swap changes every layout, and that intersecting the pair's footprints names candidate resources rather than a proven cause. In the exhaustive prevalence tables every midpoint is already known, so localization costs nothing there; for requirement flips and beyond the exhaustive range, \sys{} runs the bisection on the two witness layouts. The resources the pair shares classify the sensitivity under the five labels of Section~\ref{sec:study}, and for a learner they are the part of the verdict to act on.

\subsection{Implementation}

The exploration, the lenses, the footprint extractor, the reduction, and the attribution are about four thousand lines of JavaScript and Python. A regression suite of \nBenchPrograms{} labeled programs (hand-built control-and-mutant pairs and a generated benchmark with intended values) passes on every verdict in both configurations.

\section{Study Design}
\label{sec:study}

We ask five questions about the phenomenon and answer each by running \sys{} on real programs.

\begin{itemize}
\item \textbf{RQ1 (Prevalence).} How many real programs behave differently under some realizable stacking order, at which lens, and over which horizon?
\item \textbf{RQ2 (Severity).} Within a sensitive program, what fraction of the realizable orders produce the other behavior, and how many random re-runs would it take to notice?
\item \textbf{RQ3 (Consequence).} Does the stacking order change whether a student's program meets the requirements of its assignment?
\item \textbf{RQ4 (Resources).} Which resources do the racing sprites share, and how often is the random stream their highest-priority shared resource?
\item \textbf{RQ5 (Comparison).} Do a static bug-pattern checker and random re-execution, the two checks a teacher could apply without a model of the scheduler, see what \sys{} sees?
\end{itemize}

\subsection{Corpora}

We use three corpora of real programs; no program was written for this study. A program has a \emph{scheduling choice} when at least two sprites are in the activation closure of Section~\ref{sec:tool} under the fixed input ($k \geq 2$).

\runin{Course assignment set} Student solutions to \nCourseAssignments{} assignments of an introductory Scratch course, each a small game or animation with a written task description (for example, a shark that eats fish, a basket that catches eggs, a sight that shoots bats). Each submission is one student's complete solution to one assignment, most with two to five contributing sprites. We use the \nCourseTotal{} distinct programs (byte-identical copies counted once), \nCourseConc{} of which have a scheduling choice; the set serves RQ1 to RQ5.

\runin{Public random sample} \nPublicTotal{} projects drawn uniformly at random from the identifier space of the public Scratch website in June 2026, keeping every identifier that resolved to a shared project. The sample was rebuilt with each project's original costume and sound assets, so that costume geometry, which touch tests depend on, is the real one. Projects that reference extensions the headless VM cannot run (music, speech, translation, video) had those blocks replaced by no-ops (\nPublicStubExtstrip{} of the \nPublicTotal{} projects). \nPublicLoaded{} projects load, \nPublicConc{} of them with a scheduling choice; the sample serves RQ1, RQ2, RQ4, and RQ5.

\runin{Online-judge submissions} \nJjcTotal{} submissions by \nJjcUsers{} learners to \nJjcProblems{} problems of a Scratch online judge for school-age students, used for RQ3 because each problem comes with a precise statement of requirements. The \nJjcProblemsUsed{} problems with at least four submissions that have a scheduling choice enter the analysis: \nJjcStudied{} submissions by \nJjcStudiedUsers{} learners, of which \nJjcKtwo{} (by \nJjcKtwoUsers{} learners) have a scheduling choice. A learner may submit several times, and RQ3 counts submissions.

\subsection{Procedure}

\runin{Input and horizon} For RQ1, RQ2, RQ4, and RQ5 the input is the green flag alone. The horizon is 30 ticks (one second at 30 frames per second), the window in which a program's initialization runs; we repeat RQ1 at 90 and 300 ticks on the course set and with two further random seeds. For RQ3 the input is a scripted interaction per assignment: key presses and clicks on a fixed timeline, and an answer to every question the program asks. The horizon is a play length chosen per assignment, between 8 and 35 seconds.

\runin{Verdicts} We explore a program in one of three ways. (i) With $k \le 5$ contributing sprites we run every order. (ii) With $k > 5$ we run one order per class of the conditional reduction when the conflict graph admits at most 1{,}024 classes; for a complete graph this is again every order. (iii) Otherwise we run a sample of five orders and report the verdict as a lower bound. A program whose every order ran, under (i) or under (ii) with a complete graph, has a \emph{complete table}; RQ2 and the audits use complete tables. Every program is run twice under the identity order, and a program whose two runs differ is excluded as nondeterministic. Headline numbers use the final-state lens and the bounding-box configuration; we repeat the prevalence runs and both requirement studies in the pixel configuration and report where the two disagree (RQ1, RQ3).

\runin{Statistics} We report each headline rate with a Wilson 95\% interval. For the public sample it describes sampling from the site's projects; the course and online-judge sets are complete populations, for which the interval is a reference for reading the set as a sample of similar student work, and it does not account for clustering by learner and assignment.

\runin{Requirement predicates (RQ3)} For every assignment and problem with at least two submissions that have a scheduling choice (\nReqProblems{} in all: 12 course assignments and the 7 online-judge problems) we wrote between three and six requirement predicates. Each predicate is an operational reading of one sentence of the task description and maps the recorded run to \emph{holds}, \emph{fails}, or \emph{undetermined}. A conditional requirement (a hit sprite says \emph{ah}) is evaluated only when its condition is observed in the run. A requirement whose timing the statement fixes (the completion message comes after the last virus) fails when its consequence comes early. A requirement whose condition the trace cannot observe (a contact between sprites) can be confirmed by one observed consequence but not refuted. Of the \nReqPredicates{} predicates, \nReqPredConfirmOnly{} are of this kind or fix a timing, \nReqPredGated{} return undetermined when their observable condition does not occur, and \nReqPredUnconditional{} are decided in every completed run. The undetermined value therefore covers a condition that did not occur, a condition the trace cannot observe, and an unresolved role alike. For example, ``the robot moves when an arrow key is pressed'' becomes a predicate over the robot's position during the key-press windows. Sprites are bound to the roles a statement names (the robot, the basket) from the program itself: the sprite that reads the arrow keys is the player, and the sprite that creates clones is the spawner. Section~\ref{sec:threats} discusses the predicates' validity. A submission has a \emph{requirement flip} when some predicate holds under one realizable order and fails under another, and a \emph{determinacy change} when some predicate is determinate under one order and undetermined under another.

\runin{Attribution (RQ4)} We attribute every sensitive program to a racing pair and its shared resources as described in Section~\ref{sec:tool}, and label it by the highest-priority resource type in the pair's shared footprint. In an exhaustive table the pair is the one whose adjacent swap flips the outcome most often, ties broken by saved order; beyond it, the bisection runs between the two orders of the recorded witness. The labels are five. (i) \emph{Shared variable}: the diverging observation includes a variable the pair contends for, or one sprite sets a variable the other reads. (ii) \emph{Sensed pose}: the pair's conflict includes a touch test or a read of the other sprite's pose. (iii) \emph{Clone family} and (iv) \emph{broadcast channel}: the pair contends for one. (v) \emph{Random stream}: the random stream is the pair's highest-priority shared resource. A pair labeled by the random stream may also share a counter that both increment, the timer, or an unmodeled block; Section~\ref{sec:results} reports how often.

\runin{Baselines (RQ2, RQ5)} We compute the random baseline exactly from the exhaustive tables. For a sensitive program with outcome classes of sizes $c_1,\dots,c_m$ among $k!$ orders and a budget $N$, let $n = \min(N, k!)$; sampling $n$ distinct orders misses the difference with probability $\sum_i \binom{c_i}{n} / \binom{k!}{n}$. The static baseline is LitterBox~\cite{fraser2021litterbox} version 1.12 with all bug-pattern detectors enabled, scored (i) on any detector and (ii) on the subset of detectors that concern initialization, messaging, and cloning.

\section{Results}
\label{sec:results}

Each subsection answers one question of Section~\ref{sec:study} and ends with its answer.

\subsection{RQ1: How common is it?}

\begin{table}[t]
\centering
\caption{Prevalence of schedule-sensitivity under the green flag over 30 ticks. About one program in four with a scheduling choice in the course set and one in five in the public sample changes its final state under some realizable stacking order, and finer lenses add few programs. \emph{Reduced $=$ exhaustive}: enumerable projects on which the reduction reproduced the outcome set of the full enumeration (it never missed one). \emph{Reduction saved runs}: enumerable projects with fewer than $k!$ classes.}
\label{tab:prevalence}
{\tblfont\setlength{\tabcolsep}{\tblsep}\begin{tabular}{@{}>{\raggedright\arraybackslash}p{\tblLabelWide}rr@{}}
\toprule
 & Course & Public \\
\midrule
Projects & 145 & 250 \\
Loaded and repeatable & 145 & 249 \\
With a scheduling choice ($k\geq 2$) & 121 & 121 \\
Exhaustive or reduced verdict & 109 & 97 \\
Sampled lower bound & 12 & 24 \\
Schedule-sensitive, final-state lens & 32 (26.4\%) & 22 (18.2\%) \\
\quad 95\% Wilson interval & 19.4--34.9\% & 12.3--26.0\% \\
\quad frame lens & 38 (31.4\%) & 27 (22.3\%) \\
\quad monitor lens & 38 (31.4\%) & 28 (23.1\%) \\
\quad full-trace lens & 38 (31.4\%) & 28 (23.1\%) \\
Reduced $=$ exhaustive & 101/101 & 81/81 \\
Reduction saved runs & 13/101 & 45/81 \\
\bottomrule
\end{tabular}
}
\end{table}

Table~\ref{tab:prevalence} gives the headline numbers. Of the \nCourseTotal{} course programs, \nCourseConc{} have a scheduling choice; the rest have at most one contributing sprite, hence a single layout, and are robust by construction. Of these \nCourseConc{}, \nCourseSens{} (\pCourseSens\%; 95\% CI \ciCourseSens\%) end in a different final state under some realizable stacking order. In the public sample, \nPublicConc{} of \nPublicLoaded{} loadable projects have a scheduling choice and \nPublicSens{} (\pPublicSens\%; 95\% CI \ciPublicSens\%) are sensitive. Among the programs with a scheduling choice, \nCourseExact{} course programs and \nPublicExact{} public projects received an exhaustive or conditionally reduced verdict. Every robust course verdict has a complete table, and \nPublicReducedOnlyRobust{} public robust verdicts rest on the reduction alone. The remaining \nCourseSampled{} course programs and \nPublicSampled{} public projects have conflict graphs too dense to enumerate. We sampled five orders each; \nCourseSampledConcSens{} and \nPublicSampledSens{} of them were already sensitive, and the others are unresolved rather than robust, so their contribution to the rate is a lower bound. Of the \nPublicStubExtstrip{} public projects whose extension blocks were replaced by no-ops, \nPublicExtstripConc{} have a scheduling choice and \nPublicExtstripSens{} is sensitive. Without them the public rate is \nPublicSensNoExtstrip{} of \nPublicConcNoExtstrip{}.

\runin{Lenses} We observe that finer lenses add few programs: \nCourseSensFull{} course programs and \nPublicSensFull{} public projects differ at the full-trace lens, against \nCourseSens{} and \nPublicSens{} at the final-state lens.

\runin{Horizon and seeds} At 90 and 300 ticks the count of sensitive course programs moves from \nCourseSens{} to \nCourseSensNinety{} (95\% CI \ciCourseSensNinety\%) and \nCourseSensThreeH{} (\ciCourseSensThreeH\%), from \pCourseSens\% to \pCourseSensThreeH\% of the concurrent programs; the reduction agreed with enumeration at every horizon. Two other random seeds give \nCourseSensSeedTwo{} and \nCourseSensSeedThree{} sensitive programs against \nCourseSens{}. \nCourseSeedAgree{} of the \nCourseConc{} concurrent programs keep their verdict under all three seeds.

\runin{Certificates} The reduction reproduced the outcome set of the enumeration on every enumerable project (\nCourseZeroMiss{} course, \nPublicZeroMiss{} public) at all four lenses, and its final-lens classes were homogeneous on every one (Section~\ref{sec:tool}). It cut the number of runs below $k!$ on \nCourseReducedBelow{} of \nCourseReducible{} enumerable course programs and \nPublicReducedBelow{} of \nPublicReducible{} public projects. The median factor over all enumerable projects, including those with no saving, is \xCourseRedMedian{} and \xPublicRedMedian{}. The two prevalence studies cost \nRunsTotal{} VM runs. With two VM processes per program on one workstation, the median concurrent course program takes \xCostCourseMedianS{} seconds of wall time at 30 ticks and \xCostCourseThreeHMedianS{} at 300; the slowest, a six-sprite program with a complete conflict graph, takes \xCostCourseMaxS{} seconds.

\runin{Bounding boxes versus pixels} We repeated the runs in the pixel configuration of Section~\ref{sec:tool}. Of the projects the bounding-box configuration trusts, it completed all but the one it refuses. On all \nLaneAllCompared{} projects both configurations completed (\nLaneAllFullCompared{} of them with a scheduling choice), the final-state verdicts at 30 ticks agree, and the same \nLaneCourseNewSens{} course programs and \nLanePublicNewSens{} public projects are sensitive under both. At the full-trace lens the two differ on \nLaneAllFullChanged{} of \nLaneAllFullCompared{}. Over longer play they differ on \nLaneCourseNinetyChanged{} of \nLaneCourseNinetyCompared{} course verdicts at 90 ticks and \nLaneCourseThreeHChanged{} of \nLaneCourseThreeHCompared{} at 300, every one from robust under bounding boxes to sensitive under pixels.

\takeaway{Answer to RQ1.}{In the 30-tick green-flag experiment, about one concurrent program in four (course) and one in five (public) is schedule-sensitive at the final-state lens, a lower bound while sampled projects stay unresolved. In the course set the rate grows with play time, finer lenses add few programs, and the reduction agreed with enumeration on every program where both were run.}

\subsection{RQ2: How severe is it?}

A program can be sensitive because one order in a hundred and twenty produces another outcome, or because half of them do. The complete tables answer this exactly. We find that \nBinaryCourse{} of the \nDetectCourseN{} sensitive course programs with a complete table have exactly two outcome classes of half the orders each, a coin flip that the saved layout decides. The rest have up to 120 outcomes, and their saved layout's outcome recurs under as few as one order in a hundred. The median number of outcomes is \xMedOutcomesCourse{}, and the class of the median program's saved layout holds \xMedSavedClassCourse\% of the orders (smallest class \xMedMinClassCourse\%). In contrast, in the public sample \nBinaryPublic{} of \nDetectPublicN{} are coin flips, and the medians are \xMedOutcomesPublic{} outcomes and a saved class of \xMedSavedClassPublic\%.

The consequence for testing is two-sided. On the one hand, because the split is coarse, a tester who runs a sensitive program under up to two distinct random layouts notices the difference with probability \pDetectTwoCourse\%, averaged over the \nDetectCourseN{} sensitive course programs with a complete table (\nDetectCourseTwoSprite{} of which have only two layouts); under three layouts the probability is \pDetectThreeCourse\%, and under five nearly one. On the other hand, what random re-running cannot do is certify the \nCourseExactRobust{} course programs that complete exploration finds robust: agreement on a sample, without verified coverage, does not prove that no order disagrees.

\takeaway{Answer to RQ2.}{Among the programs with a complete table, a sensitive one reproduces its saved outcome under a median \xMedSavedClassCourse\% (course) and \xMedSavedClassPublic\% (public) of its orders, so two or three random distinct layouts expose it. The complete exploration is what turns the certified majority from ``not observed to fail'' into ``no admissible layout changes the outcome.''}

\subsection{RQ3: Does it change whether a program meets its assignment?}

\begin{table}[t]
\centering
\caption{Requirement flips under the assignments' scripted inputs. A submission flips when a requirement predicate written from the assignment statement holds under one realizable stacking order and fails under another. It has a determinacy change when a predicate is determinate under one order and undetermined under another. The last two rows give the share of the flipped submissions that each lens flags.}
\label{tab:reqflip}
{\tblfont\setlength{\tabcolsep}{\tblsep}\begin{tabular}{@{}>{\raggedright\arraybackslash}p{\tblLabelWide}rr@{}}
\toprule
 & Course set & Online judge \\
\midrule
Submissions with a scheduling choice ($k\geq 2$) & 128 & 299 \\
Requirement flips & 0 (0.0\%) & 5 (1.7\%) \\
\quad 95\% Wilson interval & 0.0--2.9\% & 0.7--3.9\% \\
\quad pair shares program state & 0 & 3 \\
\quad pair shares the random stream first & 0 & 2 \\
Determinacy changes only & 10 & 7 \\
Flips flagged, final-state lens & -- & 100\% \\
Flips flagged, full-trace lens & -- & 100\% \\
\bottomrule
\end{tabular}
}
\end{table}

We report the requirement study in Table~\ref{tab:reqflip}. Among the \nSfKtwo{} course submissions with a scheduling choice, some predicate holds under one stacking order and fails under another for \nSfFlip{} of them (\pSfFlip\%; 95\% CI \ciSfFlip\%). Among the \nJjcKtwo{} online-judge submissions, the same is true for \nJjcFlip{} (\pJjcFlip\%; 95\% CI \ciJjcFlip\%). A further \nSfCovFlip{} course and \nJjcCovFlip{} online-judge submissions have a determinacy change without a requirement flip: a predicate is determinate under one order and undetermined under another, most often a contact-conditioned requirement whose consequence appears under one order only. The pixel configuration gives the same flip verdict on \nLaneSfFlipAgree{} of \nLaneSfBoth{} course and \nLaneJjcFlipAgree{} of \nLaneJjcBoth{} online-judge submissions; the \nLaneReqFlipLost{} disagreements are D and D$'$ below, whose rocket touches a rock under pixels in every sampled order. We find that for each flipped submission at least one requirement judgment depends on the saved stacking order; \nReqFlipOtherFail{} of the \nReqFlipAll{} also fail another requirement under every explored order, so the overall grade need not change.

We read all five flips by hand. (A) In the robot problem, with the virus in front the viruses appear and the robot moves under the keys, and with the robot in front neither holds. (B) In a second submission to the same problem, the completion message comes after the last virus with the virus in front and at the first tick with the robot in front. (C) In the rocket problem, with the rock in front rocks keep spawning and two are visible at once, and with the rocket in front neither holds. (D, D$'$) In one learner's two revisions of the rocket game, an outcome message appears when rock~1 precedes rock~4, the adjacent pair certified by bisection, and not when the two are swapped, where the rocket crosses the field and says nothing; this predicate holds when any sprite says an outcome word, fails when the rocket has travelled to the right edge without one, and is undetermined otherwise. A and C are the cases of Section~\ref{sec:implications}, read from the code and both runs; B, D, and D$'$ we verified from the traces of both runs. Attributing each flip to its racing pair, we find two situations. In \nReqFlipLogicAll{} of the \nReqFlipAll{} flips the pair shares program state: a score variable that one sprite resets and the other tests, or a clone family that one sprite creates and the other tests for contact. These are races on program state, and Section~\ref{sec:implications} reads two of them off the code. In the other \nReqFlipRandomAll{} the pair's highest-priority shared resource is the random stream: the order decides where randomly placed sprites start, and with them whether the scripted play reaches the requirement's consequence in time. Under the identity order \nReqIdentTrue{} requirement evaluations hold, \nReqIdentFalse{} fail, and \nReqIdentUnknown{} are undetermined.

\takeaway{Answer to RQ3.}{For \pReqFlipAll\% of the student submissions with a scheduling choice, a predicate written from a stated requirement holds under one stacking order and fails under another. The certified pair shares program state in \nReqFlipLogicAll{} of the \nReqFlipAll{} flips and the random stream first of all in \nReqFlipRandomAll{}; \nLaneJjcNewFlip{} of the \nReqFlipAll{} also flip under pixels. For \nReqCovFlipAll{} further submissions, a predicate's determinacy differs between orders. The final-state lens flagged every flipped submission in the same runs; of the submissions it flags, \rSfFlipGivenFlag\% (course) and \rJjcFlipGivenFlag\% (online judge) flip a requirement.}

\subsection{RQ4: What do the racing pairs contend for?}

\begin{table}[t]
\centering
\caption{Attribution of every sensitive program: one certified racing pair per program, labeled by the highest-priority resource its two footprints share. \emph{Shared variable}: one sprite assigns it and another reads or updates it, or the outcomes differ in it. \emph{Sensed pose}: one sprite tests or reads the other's pose. \emph{Clone family} and \emph{broadcast channel}: one sprite creates or sends, the other tests or receives. \emph{Random stream}: the pair's highest-priority shared resource is the random stream (the text reports what else those pairs share). We count each program once, by the most frequent flipping pair of its exhaustive table or by its bisection-certified pair.}
\label{tab:patterns}
{\tblfont\setlength{\tabcolsep}{\tblsep}\begin{tabular}{@{}>{\raggedright\arraybackslash}p{\tblLabelWide}rrr@{}}
\toprule
Highest-priority shared resource & Course & Public & All \\
\midrule
Shared variable & 7 & 7 & 14 \\
Sensed pose & 4 & 4 & 8 \\
Clone family & 0 & 1 & 1 \\
Broadcast channel & 1 & 3 & 4 \\
Random stream (highest priority) & 20 & 7 & 27 \\
Other (timer, pen, unmodeled block) & 0 & 0 & 0 \\
\midrule
Sensitive programs attributed & 32 & 22 & 54 \\
\bottomrule
\end{tabular}
}
\end{table}

We classify every sensitive program by the highest-priority resource its certified racing pair shares (Table~\ref{tab:patterns}). We read the pair off the exhaustive table where one exists and certify it by the bisection of Theorem~\ref{thm:localize} otherwise. Bisection certified an adjacent pair for all \nLocPrevRows{} sensitive prevalence verdicts not read off an enumerated table (\nLocPrevRuns{} runs) and for the \nLocReqRows{} requirement flips beyond it (\nLocReqRuns{} runs). In \nLocReqTruthEdge{} of the \nLocReqAllRows{} requirement flips the certified swap changes the predicate between holds and fails; in \nLocReqDetEdge{} it changes determinacy instead. No certified pair had conflict-free footprints, which would have exposed a gap in the footprint model. Four labels name program state. A \emph{shared variable} that one sprite assigns and another reads or updates, the resource of the introduction's robot, is the pair's highest-priority shared resource in \nPatVarAll{} programs. A \emph{sensed pose}, which one sprite tests for contact or reads while the other sprite moves it, is in \nPatSenseAll{}. A \emph{clone family} that one sprite creates and the other tests or creates is in \nPatCloneAll{}, and a \emph{broadcast channel} that one sprite sends on and the other receives is in \nPatBcastAll{}. Each label names the state the two footprints contend for; the mechanism behind a particular program we read only for the two cases of Section~\ref{sec:implications}. The fifth label is the random stream, the highest-priority shared resource of the certified pair in \nPatRandAll{} programs (\pPatRandAll\%; 95\% CI \ciPatRandAll\%); the swap can change which sprite draws first, a candidate cause rather than a proven one. We note that the label is a priority, not an exclusion: \nPatRandPure{} of these pairs share nothing else, \nPatRandWithCounter{} also share a counter that both increment, \nPatRandWithTimer{} the timer, and \nPatRandWithOpaque{} a block the footprint table does not model.

The table counts one pair per program with that pair's highest-priority label, so its distribution depends on that selection. The random label is more common beyond the complete tables (\nPatRandNonexhAll{} of \nSensNonexhAll{} programs) than within them (\nPatRandExhAll{} of \nSensExhAll{}), and program size and the pair-selection rule vary together there. In the course set the random label dominates (\nPatRandCourse{} of \nCourseSens{}), consistent with assignments that spawn families of sibling sprites that each draw a starting position.

\takeaway{Answer to RQ4.}{The selected pairs split evenly: \pPatLogicAll\% share program state first of all (a variable, a sensed pose, a clone family, or a broadcast channel) and \pPatRandAll\% the random stream. A checker's report should say which.}

\subsection{RQ5: Do existing analyses see it?}

We compare \sys{} with a static checker and with random re-execution, scored against \sys{}'s verdict on the programs with a resolved one (\nLbPosCourse{} witnessed sensitive and \nLbNegCourse{} robust course programs, the robust ones all with complete tables; \nLbPosPublic{} sensitive and \nLbNegPublic{} robust public projects, \nPublicReducedOnlyRobust{} of the latter by the conditional reduction). LitterBox, run with all of its bug-pattern detectors, flags \rLbRecAnyCourse\% of the sensitive course programs and \rLbRecAnyPublic\% of the sensitive public projects. However, it flags a similar share of the robust ones, for a precision of \rLbPrecAnyCourse\% and \rLbPrecAnyPublic\% (\rLbPrecAnyExhPublic\% with the public negatives restricted to complete tables). Restricted to its detectors for initialization, messaging, and cloning (the shapes closest to the four program-state labels), recall falls to \rLbRecConcCourse\% and \rLbRecConcPublic\% at a precision of \rLbPrecConcCourse\% and \rLbPrecConcPublic\%. The static patterns and the dynamic verdict are views of different things. For example, a missing initialization is a race only if another sprite can read the variable before the initialization runs, and a race often involves blocks no smell pattern names. We note that LitterBox's detectors also target other bug classes, so its precision here measures how well its warnings discriminate schedule-sensitivity, not its false-alarm rate on its own task. Two uniformly sampled distinct layouts expose a sensitive program with mean probability \pDetectTwoCourse\% over the sensitive course programs with a complete table and \pDetectTwoPublic\% over the public ones, computed from those tables; a disagreement is a witness and never a false alarm, and agreement short of full coverage certifies nothing.

\takeaway{Answer to RQ5.}{Static warnings see most sensitive programs but separate them from the robust majority poorly. Random and systematic exploration alike certify robustness only with verified coverage of the domain or of a sound quotient of it.}

\section{Two Cases, and What To Do About Them}
\label{sec:implications}

\runin{Two requirement flips, read from the code} Both cases are online-judge submissions from RQ3 (Table~\ref{tab:reqflip}); we read the mechanisms off the blocks and the witness runs. Both turn on a fact about Scratch files that compounds the scheduling one: a saved project stores the current value of every variable and the current position of every sprite. Whatever the last run left behind is the state the next run starts from.

The first is the robot of Figure~\ref{fig:robot}, whose mechanism Section~\ref{sec:background} explains. Under the robot-first order two requirements flip: the viruses never appear, and the robot does not respond to the arrow keys. The completion message, which the statement ties to the elimination of all viruses, is spoken with nothing eliminated, which the predicate counts as a failure. Under the other order the scripted play does not reach the end of the game, so there the requirement is undetermined. \nJjcPrematureDone{} further submissions to this problem announce completion prematurely under one order in the same way.

The second problem asks for a rocket that crosses an asteroid field: a rock sprite moves to the right edge, hides, and spawns clones that drift left, and the rocket says \emph{failure} and stops if it touches a rock. In one submission the file stores the original rock sprite near the left edge, at the rocket's starting position, where a previous run left it. With the rock in front, its script moves it to the right edge and hides it before the rocket's first contact test. With the rocket in front, the rocket's first loop iteration tests contact against a rock that has not moved yet, finds it, says \emph{failure} for two seconds, and then stops everything. By then the rock's other script has created one clone, at the one-second mark, and no more. The requirements that rocks keep spawning and that at least two be visible at once hold under one order and fail under the other.

In both cases \sys{}'s attribution names the racing pair, the two layouts, and the resources the pair's footprints share. These are the score variable and the robot's position in the first case, and the rock's position and its clone family in the second. The fix in each case is one block: the robot's script sets the score to zero before its loop, and the rocket's script waits one frame before its first contact test. Under the same exploration the repaired robot game is robust. The repaired rocket game meets the same requirements under both orders. However, its final state still differs, because the two orders hand different random draws to the rocks. Both submissions avoid the illustrated startup failure under their saved order.

\runin{For learners} The concept a learner needs is small. Sprites take turns, the front one goes first, and a script should not assume that another sprite's script has already run. Two habits follow: (i) initialize shared state in one script and start the scripts that depend on it from that script, with a broadcast after the initialization; and (ii) do not test contact with another sprite, or read its position, before that sprite's own script has placed it.

\runin{For teachers and automated graders} Any grading that runs a project runs it under the one stacking order saved in the file. RQ3 says that under scripted play the saved order decides at least one requirement predicate for \pReqFlipAll\% of submissions with a scheduling choice (\nLaneJjcNewFlip{} of the \nReqFlipAll{} also under pixel sensing). Three measures follow, in increasing cost. (i) Record the stacking order the grader ran under, so that a disagreement with the student is diagnosable. (ii) Run the project under a handful of distinct orders, which by RQ2 exposes most sensitive programs with a complete table, and treat a disagreement as a reason to look rather than as a verdict. (iii) Run \sys{}, which certifies the enumerable majority, exhaustively or conditionally on the graph, and for a sensitive submission a pair whose swap changes the outcome.

\runin{For the Scratch platform} The mechanism is a design choice of the platform, and the platform can change it. One change is documentation: state that scripts start in stacking order and that dragging a sprite changes that order. A second is an indication in the editor when a project's behavior depends on the order, which the exploration of this paper can compute for the enumerable majority; its latency and usefulness inside an editor remain to be evaluated. The most thorough change is to make the start order independent of the drawing order, for instance by starting sprites in the order in which they were created.

\section{Limitations and Threats to Validity}
\label{sec:threats}

\runin{What a witness is} Theorem~\ref{thm:realizable} ties every explored order to a stacking order, and the harness realizes an order by rewriting layer numbers and running the unmodified VM. Every witness is therefore two saved layouts run by the VM in the stated configuration. The reduction's soundness depends on the adequacy condition of Section~\ref{sec:tool}, which we have not proved for the extractor in general. A sprite with an unmodeled block is marked opaque and conflicts with every other sprite, so such a block disables reduction rather than enabling it. Which sprites the domain includes rests on the extractor's activation rules for broadcasts, backdrop switches, clones, and timers. A sprite those rules miss is outside every certificate. Every enumerable instance passes the class-homogeneity check at the final-state lens and the 30-tick horizon (all \nRecheckTables{} exhaustive tables) and the outcome-set agreement between reduction and enumeration at all four lenses. Every certificate speaks about the layer coordinate of the contributing sprites alone: silent sprites, the script order inside each sprite, saved values, assets, input policy, seed, horizon, and sensing configuration are held fixed. A sampled project without a witness is unresolved, not robust, and we report the \nSampledAll{} sampled verdicts as lower bounds.

\runin{Headless execution} The bounding-box configuration runs without a renderer: its stand-in takes sprite geometry, bounds, and fencing from the real costume assets (Section~\ref{sec:tool}). However, a touch test between sprites is decided by costume bounding boxes, and a colour-touch test always fails. We compared the two sensing configurations on the prevalence and requirement experiments (RQ1, RQ3). The two agree on every 30-tick final-state verdict; they differ on a handful of verdicts over longer play and on \nLaneSfFlipChanged{} of \nLaneSfBoth{} course and \nLaneJjcFlipChanged{} of \nLaneJjcBoth{} online-judge flip verdicts. However, the two configurations share one runtime, one recording, and one input policy, which the comparison cannot check. The refused project remains unchecked, and a pixel decision at a costume's edge can differ between renderer builds; every pixel run used one machine and one build.

\runin{Input and horizon} We measure prevalence under the green flag alone, over one second and under one random seed, with three and ten seconds and two further seeds as controls on the course set (RQ1). Sensitivity that only a specific interaction exposes is not counted in RQ1; RQ3 measures it under scripted interactions, where our script realizes one reading of the statement. Endpoint sensitivity need not grow with the horizon (Section~\ref{sec:model}); in the course set no program lost a sensitive verdict at a longer horizon (\nHorizonLost{} of \nHorizonCompared{}).

\runin{Requirement predicates} The predicates of RQ3 are our reading of the assignment statements: they cover selected requirements, infer roles from simple static cues, and can misjudge an individual submission. A contact-conditioned requirement (Section~\ref{sec:study}) never flips, and a flip elsewhere in the same run may rest on a contact that one order produced and the other did not. The undetermined value covers unobserved conditions, unobservable conditions, and unresolved roles alike, so a determinacy change is not evidence that a condition did not occur. A systematic misreading of a statement changes which submissions flip, in either direction. The validity of the predicates is therefore a threat of its own. The hand checks of Section~\ref{sec:results} address it for the five flips, and the identity-order values (\pReqIdentTrue\% of the determinate evaluations hold, \nReqIdentUnknown{} are undetermined) describe the predicates' behavior without validating them.

\runin{Attribution} We label the racing pairs of RQ4 by the resource types in the racing pair's footprints and a fixed priority among them. A pair that contends for both a shared variable and the random stream is labeled by the variable in two cases: (i) when the diverging observation involves that variable, or (ii) when one sprite assigns the variable and the other reads or updates it. When both sprites only increment it, or both reset it to the same constant and otherwise only increment it, we label the pair by the random stream, unless case (i) applies; it does for \nPatVarWeakAll{} public pairs whose shared counter is among the diverging observations. The choice of pair matters more than the rule. Bisection certifies that swapping the selected pair changes the observation in one layout, not that this pair stands for the program's other racing pairs. The label distribution therefore depends on the selection rule, and Section~\ref{sec:results} lists what else the selected pairs share.

\runin{Populations} The course set is one course's assignments (all small games with families of sibling sprites), which is a plausible reason for its many random-stream labels. The public sample is uniform over successfully retrieved projects in the sampled identifier range. In contrast to the course set, it represents projects rather than users, active projects, or classroom work. A learner or an assignment contributes several programs to the course and judge sets, which the binomial intervals ignore. All results hold for one version of the production VM, and a change to how it walks its target list would require running the study again.

\runin{Data use} The student programs are used as program artifacts with permission; the study observes program behavior, not students. No student, course, or platform is identified, and the programs are not redistributed.

\section{Related Work}
\label{sec:related}

\runin{Analyses of Scratch programs} The Scratch repository is large enough to study at scale~\cite{aivaloglou2016kids}, and a body of tools analyzes learners' projects. Static checkers flag bug patterns and code smells in the block graph~\cite{fraser2021litterbox,boe2013hairball,morenoleon2015drscratch,fraedrich2020bugs}. Test generators synthesize inputs and check properties one execution at a time~\cite{stahlbauer2019whisker,deiner2023testgen,goetz2022model}. Further systems verify learner programs, generate hints, debug, and grade homework~\cite{stahlbauer2020verified,obermueller2021catnip,deiner2024nuzzlebug,johnson2016itch}. To our knowledge, none of them varies the stacking order; each runs, or reasons about, the order saved in the file. In contrast, the phenomenon of this paper lives in the dimension they hold fixed, and RQ5 finds that static warnings separate the sensitive programs from the robust ones poorly. The comparison of two Scratch programs through a lattice of observation lenses~\cite{scratchlens} supplies the lenses we use to compare one program with itself across its schedules.

\runin{Concurrency testing} Stateless model checking explores the interleavings of a concurrent program, and partial-order reduction prunes interleavings that commute~\cite{godefroid1996partial,flanagan2005dpor,abdulla2014optimal}, resting on Mazurkiewicz trace equivalence~\cite{mazurkiewicz1987trace} and on the state-explosion problem it addresses~\cite{valmari1998state}. Systematic testing of reactive systems explores message and callback orders~\cite{desai2015pingpong}. Our setting is narrower. In our fixed environment the initial layout is the parameter varied between deterministic executions, so the space is a finite symmetric group rather than an interleaving tree, and a complete exploration is feasible for the small projects of this study. The reduction we use is a static, depth-one instance of these ideas; it buys a certificate at the price of a conflict graph that we must assume sound beyond five sprites. Schedule-sensitivity is also a cousin of flaky and order-dependent tests, whose outcome depends on something the harness leaves unspecified~\cite{luo2014flaky,lam2019idflakies,gligoric2015nondex}.

\runin{Determinacy and program equivalence} Bernstein's conditions give the classical sufficient condition for two tasks to commute~\cite{bernstein1966}; robustness at the full-trace lens is the corresponding observational property. The footprint model is an abstraction in the sense of abstract interpretation~\cite{cousot1977abstract}; reduced-only verdicts assume its adequacy, and enumerable cases receive the class-homogeneity check of Section~\ref{sec:tool}. Translation validation, regression verification, differential symbolic execution, and semantic diffing decide whether two programs agree~\cite{pnueli1998translation,godlin2009regression,person2008differential,lahiri2012symdiff,ramos2011equivalence,jackson1994semanticdiff}. In contrast, we ask whether one program agrees with itself under a change the language treats as cosmetic.

\section{Conclusion}
\label{sec:conclusion}

A Scratch program's scripts start in the order in which its sprites are stacked, and a drag changes that order. The schedules this produces are the permutations of the initial stacking order, a finite space that a checker can enumerate for small programs, reduce by a conflict graph for larger ones, and sample beyond. In real programs the order decides the final state of about one concurrent program in four (course) and one in five (public) and, among those with a complete table, at most half of the orders typically reproduce the saved outcome. It decides a requirement predicate for \pReqFlipAll\% of student submissions with a scheduling choice, and half of the certified racing pairs contend first of all for the random stream. The editor does not expose the mechanism today. The documentation can name it, the editor can flag it, and the platform can take the drag out of the schedule; the rates above motivate those changes, and their effects remain to be measured. Every environment that saves a coordinate its language treats as cosmetic has a scheduler of this kind, and the same finite exploration applies to each. Three problems remain open: (i) an extractor proved adequate for the VM's activation semantics, (ii) requirement predicates validated against graders' judgments, and (iii) an editor indication whose usefulness is measured with learners.

\bibliographystyle{IEEEtran}
\bibliography{refs}

\end{document}